\pdfoutput=1
\documentclass[
onefignum,onetabnum]{siamart251216}

\usepackage[utf8]{inputenc}
\usepackage[T1]{fontenc}
\usepackage{booktabs}
\usepackage{amsfonts}
\usepackage{amssymb}
\usepackage{nicefrac}
\usepackage{microtype}
\usepackage{mathtools}
\usepackage{siunitx}
\usepackage{graphicx}
\title{
Well-posed by Design: Learning Constitutive Laws from Velocity Data using Convex Neural Network Potentials
\thanks{Version from \today.%
\funding{The authors were supported by the Multidisciplinary University Research Initiatives (MURI) Program, Office of Naval Research (ONR) grant \#N00014-19-1-2421, and by ONR \#N00014-26-1-2101.
G.S.~also appreciates support from the US National Science Foundation (NSF) under \#2343866, \#2411229 and \#2411349.}}}

\DeclareMathOperator{\tr}{tr}
\DeclareMathOperator{\dev}{dev}
\DeclareMathOperator*{\argmin}{arg\,min}

\newcommand{\RRp}{\mathbb{R}_{\geq 0}}

\newcommand{\bA}{\boldsymbol{\mathrm{A}}}

\newcommand{\bD}{\boldsymbol{\mathrm{D}}}
\newcommand{\bI}{\boldsymbol{\mathrm{I}}}
\newcommand{\bS}{\boldsymbol{\mathrm{S}}}
\newcommand{\bO}{\boldsymbol{\mathrm{O}}}

\newcommand{\bsigma}{\boldsymbol{\sigma}}
\newcommand{\btau}{\boldsymbol{\tau}}
\newcommand{\bi}{\boldsymbol{\iota}}

\newcommand{\bx}{\boldsymbol{x}}

\newcommand{\dd}{\mathrm{d}}
\newcommand{\bu}{\boldsymbol{u}}
\newcommand{\bv}{\boldsymbol{v}}
\newcommand{\bzero}{\boldsymbol{0}}

\newcommand{\RR}{\mathbb{R}}
\newcommand{\Sbb}{\mathbb{S}}
\newcommand{\Obb}{\mathbb{O}}

\newcommand{\bff}{\boldsymbol{f}}

\newcommand{\Cc}{\mathcal{C}}

\newcommand{\Fc}{\mathcal{F}}

\newcommand{\Jc}{\mathcal{J}}

\newcommand{\Lc}{\mathcal{L}}

\newsiamremark{remark}{Remark}

\author{Gonzalo G.~de Diego\thanks{Courant Institute School of Mathematics, Computing and Data Science, New York University, USA, and E.T.S.I.~Telecomunicaci\'on, Universidad Polit\'ecnica de Madrid, Spain (\email{gonzalo.gonzalez@courant.nyu.edu}, \email{g.gonzalezd@upm.es}).}
  \and Georg Stadler\thanks{Courant Institute School of Mathematics, Computing and Data Science, New York University, New York, NY, USA (\email{stadler@cims.nyu.edu}).}}

\headers{Well-posed by design: Learning constitutive laws}{G. G. de Diego and G. Stadler}

\ifpdf
\hypersetup{
  pdftitle={Well-posed by Design: Learning Constitutive Laws from Velocity Data using Convex Neural Network Potentials},
  pdfauthor={G. G. de Diego and G. Stadler}
}
\fi

\begin{document}

\maketitle

\begin{abstract}
Learning constitutive laws of complex fluids from velocity data (i.e.~indirect observations) is a PDE-constrained inverse problem in which expressive neural parameterizations risk breaking the well-posedness of the forward physics model. We address this tension by learning, rather than the constitutive law itself, the dissipation potential: a scalar function whose convexity, frame-indifference, and dissipativity propagate into the underlying continuum mechanics and guarantee the structural properties needed for a well-posed and generalizable forward PDE. We introduce ICNNE, an input-convex neural architecture that exactly enforces convexity, frame-indifference, and evenness in the second strain-rate invariant via symmetrization, alleviating the numerical instabilities that occur at small strain rates in standard formulations. By weakly enforcing zero gradients in the origin of the potential via penalization, ICNNE also satisfies the dissipativity property. The loss objective and its gradient are computed using a combination of finite element and neural network methods, coupling the Firedrake and PyTorch libraries. We evaluate on four problems spanning compressible and incompressible regimes: compressible Navier–Stokes, Herschel–Bulkley yield-stress flow, Hibler's viscous-plastic sea-ice model, and discrete-element-method data with no known constitutive law. We show that the learned potentials (i) recover ground-truth physics where it is known, (ii) transfer accurately to geometries unseen during training, and (iii) succeed in regimes where unstructured methods diverge. These results suggest that embedding mathematical structure into the parameterization, rather than into the loss, is a robust path to learning physics from indirect data.
\end{abstract}



\section{Introduction}

Constitutive laws that relate internal stress to deformation rates characterize the behavior of a wide range of complex fluids, such as cement flow \cite{sant2008}, lava \cite{griffiths2000}, Arctic sea-ice floes \cite{feltham2008}, and blood \cite{baskurt2003}. Multiple machine learning (ML) approaches exist for deriving constitutive laws from data. These approaches can be classified according to two key principles: 1) what data is used, and 2) what functional form is assumed for the constitutive law. Stress and strain-rate data, often known as direct data, are not directly observable in most cases, and one must typically rely on displacement or velocity data (i.e.\ indirect data). A common practice is to use a predefined functional form for the constitutive law (e.g., a power law) and fit a small number of scalar parameters (e.g., the coefficient and exponent of the power law) to reproduce the data. However, without a-priori knowledge, we must use a more expressive family of functions, such as polynomials or neural networks (NNs), to explain the data, yielding a function inversion  problem.

Here, we propose a \emph{function inversion} framework, summarized in \cref{fig:summary}, to infer constitutive laws for complex fluids from sparse \emph{velocity} data.
Because velocities depend indirectly on the constitutive law, they can inform the constitutive relation only when combined with the partial differential equation (PDE) that enforces the physics governing fluid motion.
An unrestricted NN parameterization of the constitutive law may break the well-posedness (i.e.,~solvability) of the underlying PDE. Therefore, to incorporate the properties required for the PDE to be well-posed, we refrain from directly learning the effective viscosity that characterizes a fluid's constitutive law. Instead, we learn the convex dissipation potential whose derivatives yield the viscosity.
We parameterize this potential with input-convex NNs while also enforcing key mathematical properties, including frame-indifference, convexity, and dissipativity. Consequently, the resulting fluid equations of motion are \emph{well-posed by design}, which is crucial for the robustness of the PDE-constrained optimization problem solved during training. From a computational standpoint, we embed the equations of motion into the training loop by coupling machine learning (ML) and finite element (FE) libraries.

\begin{figure}
    \includegraphics[width=\linewidth, trim={0 0.25cm 0 0}, clip]{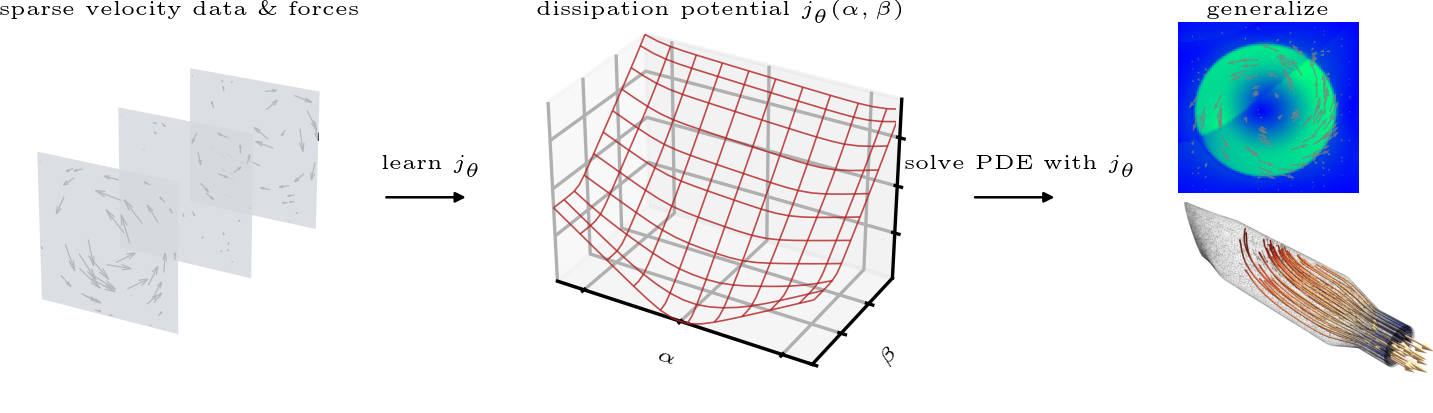}
    \caption{Summary of our approach to infer a constitutive law for a complex fluid. Using a single or multiple snapshots of sparse velocity data, we learn a dissipation potential $j_\theta$ parameterized as input-convex NNs that characterizes the constitutive relation. This relation can then be transferred to different problem settings, including problem setups with different spatial dimensions.}
    \label{fig:summary}
\end{figure}

\subsection{Related work}\label{subsec:related_work}
When learning complex fluid constitutive laws with \emph{predefined functional expressions} containing unknown parameters, current approaches either employ gradient-based optimization with the governing equations imposed as constraints \cite{akerson2025, hu2024,kontogiannis2025}, or treat these equations as penalized residuals in physics-informed neural networks (PiNNs)
\cite{lardy2025, thakur2024}. Yet, unless the material behavior is accurately characterized in advance, such predefined forms can introduce inference bias.
To reduce this bias, in \cite{saadat2022, mahmoud2024}, constitutive relations are chosen from a function library using sparsification to promote simple models.

An alternative approach, which is also our approach, is to employ NN-based parameterizations of the constitutive law.
In \cite{xu2021,lennon2023}, for instance, the focus is on learning viscoelastic constitutive laws: in \cite{xu2021}, this is done using velocity data but without enforcing key physical principles; in \cite{lennon2023}, the authors use experimental rheometer stress (rather than velocity) data.
In \cite{parolini2025}, input-convex NNs are used to parametrize the effective viscosity and learn power-law flows from stress data. While convexity, when imposed on the viscosity, restricts the type of constitutive laws, it guarantees well-posedness of the governing PDE. A class of monotone constitutive laws built with an NN are fitted to both stress and velocity data in \cite{sunol2025}.
In \cite{dediego2026}, an unknown constitutive law is inferred that reproduces synthetic sea ice Lagrangian floe velocity data in a one-dimensional setup; the same method is also used for learning incompressible fluid laws in \cite{dediego2026b}.
As in this present work, \cite{geuken2025, RajCaoStuartEtAl26} use input-convex NNs to learn \emph{potentials} in the context of elasticity (rather than fluids) and from direct, i.e.~stress-displacement, data (rather than indirect, i.e.~displacement, data).

\subsection{Contributions and limitations}\label{subsec:contributions}
Our main \emph{contributions} are:
(1) We propose learning a \emph{convex dissipation potential} that encodes the constitutive law instead of learning the law (i.e.~the effective bulk and shear viscosities) directly. We parameterize this potential with symmetrized (i.e., even) input-convex neural networks (which we call ICNNEs) that incorporate crucial physical principles that guarantee the well-posedness of the governing equations. We do this without restrictive assumptions about the constitutive law or limiting the complex fluids the approach applies to.
(2) We learn the constitutive relation from \emph{velocity data}, which entails solving a PDE- or physics-constrained learning problem during training. Computationally, this is made possible by integrating ML and FE methods and software. We exploit automatic differentiation and adjoint functionalities to compute gradients, ensuring that the (discrete) PDE constraint is exactly enforced.
(3) We consider several NN-based models with varying degrees of in-built structure and evaluate their performance on a collection of benchmark data. Our numerical results demonstrate that unstructured models are, in general, unable to faithfully capture complex fluid dynamics. Moreover, to show we capture the underlying physical properties, we demonstrate that ICNNE-learned constitutive relations generalize to entirely different geometries and even across spatial dimensions.
(4) Our approach is designed for learning compressible flow models, which have higher-dimensional inputs than incompressible ones. As far as we are aware, ML-driven approaches for inferring constitutive laws in compressible fluids have not yet been established.

The main \emph{limitations} of the method are as follows:
(1) The training is currently done only for stationary problems or snapshots of time-dependent problems in which the Stokes assumption is valid. However, there is no fundamental barrier to extending the approach to learn from time-dependent velocity data, where the Stokes assumption is not satisfied.
(2) A standardized procedure for coupling neural networks with finite-element PDE solver libraries is not yet available. Our current implementation will hopefully enable others to build upon this framework and, in the long term, help establish learning within a more accessible and standardized combined NN–PDE setting.
(3) The present framework is restricted to local, history-independent constitutive laws and therefore does not yet cover fluids with memory or nonlocal behavior. Besides technical challenges, we see no fundamental obstacles to generalizing our method to these broader classes of complex fluids.

\section{Equations of motion for a viscous fluid}

Let $\Omega \subset \RR^2$ be the domain occupied by a viscous fluid. We limit our analysis to $\RR^2$ because we infer constitutive relations from 2D velocity fields, and the set of frame-indifferent constitutive laws in $\RR^2$ includes nearly all widely used compressible non-Newtonian models \cite{bird1987, malek2010}.
This limitation does not carry over to the computational implementation, which is valid for problems of any dimension because we only learn a local operator from strain-rates to stresses. For example, in \cref{subsec:hb}, we carry out a 3D computation with a constitutive law learned from 2D data.

The evolution of the fluid in $\Omega$ is described by conservation laws. Denoting the fluid velocity by $\bu$, the Cauchy stress tensor by $\bsigma$, the density by $\rho$, and the (possibly nonlinear) external forces by $\bff$, the equations for conservation of mass and momentum are given by
\begin{subequations}\label{eq:navier_stokes}
    \begin{align}
    \frac{\partial \rho}{\partial t} + \nabla\cdot(\rho\bu) = 0 \quad \text{in $\Omega$},& \label{eq:conservation_mass} \\
    \rho\left(\frac{\partial \bu}{\partial t} + \left( \bu\cdot\nabla\right)\bu\right) - \nabla\cdot \bsigma = \bff(\bu) \quad \text{in $\Omega$}.& \label{eq:momentum_balance}
    \end{align}
\end{subequations}
In this work, we consider only isothermal fluids, which allows us to dispense with the energy equation, because the pressure is a function of density alone.
However, aside from boundary and initial conditions, additional equations are still required to arrive at a closed system of equations. These include an equation of state that relates pressure and density, and the constitutive law, which is the subject of the next section.

\subsection{The constitutive law}\label{subsec:constutitive-law}

A key element of a viscous fluid model is the constitutive law, also called the rheology. It is generally written as a relationship between the Cauchy stress tensor $\bsigma$ and the strain-rate tensor $\bD\bu$, which is defined as the symmetric part of the velocity gradient $\nabla\bu$. In particular, if $\Sbb_2$ denotes the set of symmetric matrices in $\RR^{2\times 2}$, a constitutive law that is explicit and local is written in terms of a pressure field $p$ and a matrix operator $\Cc:\Sbb_2\to\Sbb_2$ as
\begin{align}
    \bsigma = -p\bI + \Cc(\bD\bu).
\end{align}
Instead of inferring a general matrix operator $\Cc$ from velocity data, we restrict its functional form by enforcing physical constraints.

\paragraph{Frame-indifference}

A fundamental principle is that of frame-indifference, which means that $\Cc$ is invariant to rotations of the reference frame \cite{truesdell2004}. In \cref{subapp:frame_indifference}, we show that the constitutive law is frame-indifferent if there exist an effective bulk viscosity $\zeta:\RR\times \RRp\to\RR$ and shear viscosity $\eta:\RR\times \RRp\to\RR$ such that
\begin{align}\label{eq:representation_of_C}
	\Cc(\bD\bu) = 2\eta(\alpha,\beta)\,\bS\bu + \zeta(\alpha,\beta) \alpha\, \bI,
\end{align}
where $\bS\bu$ is the deviatoric part of $\bD\bu$, $\alpha$ and $\beta$ correspond with two invariants of $\bD\bu$, defined as
\begin{align}\label{eq:invariants}
   \alpha:=\tr{\bD\bu} \quad \text{and} \quad \beta := |\bS\bu|,
\end{align}
and $|\cdot|$ denotes the matrix norm $|\bA|^2:=\frac{1}{2}\bA:\bA$ for $\bA\in\RR^{2\times2}$. We note that, in $\RR^3$, the general class of frame-indifferent models admits an additional term proportional to $(\bD\bu)^2$ \cite{truesdell2004}, but most fluid models ignore this higher order term \cite{bird1987, malek2010}.

\paragraph{Dissipativity}

We require viscous stresses to dissipate energy, such that $\Cc(\bD\bu):\nabla\bu \geq 0$ for all $\bu$, which implies compliance with the second principle of thermodynamics \cite{morro2023}. This condition, which we refer to as dissipativity, is equivalent to:%
\begin{align}\label{eq:dissipative_visc}
    \zeta(\alpha,\beta)\alpha^2 + 4\eta(\alpha,\beta) \beta^2 \geq 0 \quad \forall (\alpha, \beta)\in\RR\times\RRp.
\end{align}

\paragraph{Existence of a convex dissipation potential}

The final property we seek is that the constitutive law be derived from a convex and differentiable dissipation potential $\Jc$ defined in terms of a local potential $j:
\RR\times\RRp\to\RR$ as
\begin{align}\label{eq:dissipation_potential}
    \Jc(\bu) := \int_\Omega j(\alpha(\bu),\beta(\bu))\,\dd\bx.
\end{align}
In the next section, we show that the viscosities $\eta$ and $\zeta$ can be written in terms of the derivatives of $j$. In this work, we seek a constitutive law derived from $\Jc$ for three reasons.
(1) Practically all fluid constitutive laws of interest in industrial and scientific applications are derived from dissipation potentials \cite{goddard2019, morro2023}.
(2) For convex and dissipative potentials $\Jc$, we find the resulting equations of motion to be well-posed and numerically stable, see \cref{thm:well-posed}  for a theoretical justification. Well-posedness and numerical stability are crucial properties to ensure the robustness of our PDE-constrained optimization method.
(3) If $\Jc$ exists, we can learn the local dissipation potential $j$ instead of the viscosities $\zeta$ and $\eta$. This is desirable because it is easier to enforce the PDE's well-posedness, as shown below in \cref{thm:properties_j}.

\subsection{Energy formulation of the momentum balance equation}

Let $V$ denote a suitable Banach space of velocity fields on $\Omega$, with duality pairing $\langle \cdot, \cdot \rangle$.
If the dissipation potential $\Jc$ exists and is well defined and differentiable on $V$, we can formulate the momentum balance equation \eqref{eq:momentum_balance} as the minimization of an energy functional $\Fc:V\to \RR$ whenever two additional assumptions hold. First, the inertial forces in \eqref{eq:momentum_balance} must be neglected, such that $\rho (\partial\bu/\partial t + (\bu\cdot\nabla)\bu) \approx 0$. Without inertia, the momentum balance equation is called the Stokes equation. Second, we assume that the forces $\bff$ can be derived from the Fr\'echet derivative of a differentiable potential $\Phi:V\to \RR$. That is, $\langle \Phi'(\bu), \bv \rangle = \int_\Omega \bff(\bu)\cdot\bv\,\dd x$. Then, we can define the energy functional $\Fc:V\to \RR$ as
\begin{align}\label{eq:condition_on_F}
    \Fc(\bu) := \Jc(\bu) - \int_\Omega p \left(\nabla\cdot\bu\right) \,\dd x - \Phi(\bu).
\end{align}
In weak form, the Stokes equation with \eqref{eq:representation_of_C} as a constitutive law is given by
\begin{align}\label{eq:weak_form_stokes}
    \int_\Omega (-p + \zeta(\nabla\cdot\bu))(\nabla\cdot\bv) + 2 \eta\, (\bS\bu:\bS\bv)\,\dd x - \int_\Omega \bff(\bu) \cdot \bv\,\dd x = 0 \quad \forall \bv\in V.
\end{align}
If a dissipation potential exists, the Stokes equation is recovered by setting the Fr\'echet derivative of the energy functional $\Fc$ in the direction of $\bv$ to zero.  We have that
\begin{align}\label{eq:first_variation_energy}
    \begin{split}
    \langle \Fc'(\bu), \bv \rangle &=  \int_\Omega \frac{\partial j}{\partial \alpha} \,\tr{\bD\bv} + \frac{1}{2\beta} \frac{\partial j}{\partial \beta} \left( \bS\bu:\bS\bv \right)\,\dd x \\
    & - \int_\Omega p(\nabla\cdot\bv)\,\dd x - \int_\Omega \bff(\bu) \cdot \bv\,\dd x \quad \text{for all $\bv\in V$}.
    \end{split}
\end{align}
We can then relate the Cauchy stress tensor $\bsigma$ to the local dissipation potential $j$ by using the fact that $ \langle \Fc'(\bu), \bv \rangle = 0$ is equivalent to \eqref{eq:weak_form_stokes}. This yields
\begin{align}\label{eq:relationship_j_viscosity_2}
    \frac{\partial j}{\partial \alpha}(\alpha,\beta) = \zeta(\alpha,\beta) \alpha \quad \text{and} \quad \frac{\partial j}{\partial \beta} (\alpha,\beta) = 4\eta(\alpha,\beta) \beta.
\end{align}
Therefore, the viscous stresses $\Cc(\bD\bu)$ can be written in terms of $j$ as
\begin{align}\label{eq:viscous_stresses_j}
    \Cc(\bD\bu) = \frac{\partial j}{\partial \alpha}(\alpha,\beta)\, \bI + \frac{1}{2\beta}\frac{\partial j}{\partial \beta} (\alpha,\beta)\,\bS\bu.
\end{align}

A crucial element in our approach to learning constitutive laws is the well-posedness of the underlying fluid equations of motion. The following theorem presents sufficient conditions for $\Fc$ to have a unique minimizer.

\begin{theorem}\label{thm:well-posed}
    For a smooth and connected domain $\Omega\subset\RR^d$, with $d\geq1$, let $V$ be a closed subspace of the Sobolev space $W^{1,r}(\Omega)$ with $r > 1$. Assume that the energy functional $\Fc:V\to\RR$ is convex and coercive in the following sense: for any sequence $\bu_n$ in $V$ such that $\|\bu_n\|_V \to \infty$, it follows that $\Fc(\bu_n)\to\infty$. Then there exists a minimizer of $\Fc$ in $V$. Moreover, if $\Fc$ is strictly convex, this minimizer is unique.
\end{theorem}

\begin{proof}
    A proof of the existence of a minimizer can be found in \cite[Sec 8.2]{evans2010}. Regarding the uniqueness of minimizers, assume $\bu_1$ and $\bu_2$ are two different minimizers of $\Fc$. In this case, $\bu_1\neq\bu_2$ yet $\Fc(\bu_1) = \Fc(\bu_2)$. Then, if $\Fc$ is strictly convex, we have that
    \begin{align}
        \Fc\left( \frac{1}{2}\bu_1 + \frac{1}{2}\bu_2\right) <
        \frac{1}{2}\Fc(\bu_1) + \frac{1}{2}\Fc(\bu_2) = \Fc(\bu_1).
    \end{align}
    This is a contradiction because $\bu_1$ and $\bu_2$ minimize $\Fc$.
\end{proof}

Theorem \ref{thm:well-posed} establishes two  conditions for the equations of motion to be well-posed: coercivity and convexity of $\Fc$. Coercivity depends on the function space $V$ on which $\Fc$ is defined and, often, on the boundary conditions and applied forces. Identifying the space $V$ is a complex question and, for many parameterizations of $\Jc$, we cannot expect $V \subset W^{1,r}(\Omega)$ to hold for $r > 1$. This, for example, is the case for the viscous plastic sea ice model in \cref{subsec:vp}, see e.g.~\cite{brandt2022}.

\begin{remark}\label{remark:extension-well-posedness}
A possible avenue for extending the validity of \cref{thm:well-posed} to $r=1$ could involve the space of vector fields of bounded deformation in $\Omega$, denoted by $BD(\Omega)$ \cite{temam1985}. This space consists of the set of vector fields in $L^1(\Omega)$ whose symmetric gradients belong to the space of matrix-valued Radon measures. By relaxing the definition of the energy functional to incorporate measure-valued derivatives, one can show that, under linear growth conditions, the energy functional is coercive and, via a compactness argument, prove the existence of a minimizer. This extension will be explored in future work.
\end{remark}

\subsection{Sufficient conditions for a structure-preserving constitutive law}

In section \ref{sec:computational_framework} below, we present a computational framework for inferring constitutive laws that satisfy the three properties listed in \cref{subsec:constutitive-law}. To facilitate the construction of finite dimensional manifolds of candidate functions automatically satisfying these properties, we establish three properties on $j$ for $\Jc$ to be convex and the resulting constitutive law dissipative (frame-indifference holds due to $j$'s dependence on $\alpha$ and $\beta$). In practice, as demonstrated by the numerical examples in section \ref{sec:numerics}, guaranteeing convexity and dissipativity is enough for achieving robust models for complex fluids. This may be justified by the potential extension of \cref{thm:well-posed} to the case where $r=1$, which we explain in Remark \ref{remark:extension-well-posedness}. Under these conditions, we could establish the existence of solutions to energy functionals with linear growth conditions. In \cref{subsec:structure_preserving_architecture}, we explain that linear growth conditions arise in the ICNNE model due to the use of softplus activation functions.

\begin{theorem}\label{thm:properties_j}
Let the constitutive law be defined in terms of a differentiable local dissipation potential $j=j(\alpha,\beta)$ via \eqref{eq:viscous_stresses_j}. Assume that the following properties hold: 1) $j$ is convex, 2) $j$ is increasing in $\beta$, and 3) the gradient $\nabla_{\iota}j := (\frac{\partial j}{\partial \alpha}, \frac{\partial j}{\partial \beta})^\top$ is zero when $(\alpha,\beta) = (0,0)$. Then, the constitutive law is frame-indifferent, dissipative, and is derived from the convex dissipation potential $\Jc$ defined in \eqref{eq:dissipation_potential}.
\end{theorem}

\begin{proof}
    Let $j:\RR\times\RRp\to\RR$ be convex, in the sense that
    \begin{align}
        j(s\bi_1+(1-s)\bi_2) \leq s j(\bi_1) + (1-s)j(\bi_2)
        \quad \forall s
        \in[0,1],\, \bi_1,\bi_2 \in \RR\times\RRp.
    \end{align}
    We also assume that $j = j(\alpha,\beta)$ is differentiable, increasing in $\beta$, and its gradient zero at the origin $(\alpha,\beta)=(0,0)$. For a domain $\Omega\subset\RR^2$, consider the dissipation potential defined by \eqref{eq:dissipation_potential}.
    The constitutive law for the model is then defined by the viscous stress tensor \eqref{eq:representation_of_C}. Since $j$ depends only on the invariants of $
    \bD\bu$, the constitutive law is rotationally-invariant, and thus frame-indifferent, as defined in appendix \ref{subapp:frame_indifference}. Next, we show that $\Jc$ is convex: let $\bu,\bv$ be two velocity fields and $s\in[0,1]$. Then,
    \begin{align}
        \Jc(s\bu + (1-s)\bv)=\int_\Omega j(s\alpha (\bu) + (1-s)\alpha(\bv),\beta(s\bu + (1-s)\bv))\,\dd x.
    \end{align}
    The function $\bu\mapsto\beta(\bu)$ is clearly convex. Moreover, since $j$ is increasing in $\beta$, it follows that
    \begin{equation}
        \begin{split}
        \int_\Omega j(s\alpha (\bu) + (1-s)\alpha(\bv),\beta(s\bu + (1-s)\bv))\,\dd x \\ \leq  \int_\Omega j(s\alpha (\bu) + (1-s)\alpha(\bv),s\beta (\bu) + (1-s)\beta(\bv))\,\dd x
        \end{split}
    \end{equation}
    Then, the convexity of $\Jc$ follows from the convexity of $j$. For the dissipativity of the constitutive law, first notice that $j$ is convex if and only if it has a monotonically increasing gradient:
    \begin{align}
        \left( \nabla_\iota j(\bi_1) - \nabla_\iota j(\bi_2) \right) \cdot (\bi_1 - \bi_2) \geq 0 \quad \forall \bi_1,\bi_2\in\RR\times\RRp.
    \end{align}
    By setting $\bi_1=(\alpha,\beta)$ and $\bi_2 = \bzero$ above, we find that
    \begin{align}\label{eq:dissipative_visc_j}
        \frac{\partial j}{\partial\alpha}(\alpha,\beta)\alpha+ \frac{\partial j}{\partial\beta}(\alpha,\beta)\beta \geq 0 \quad \forall (\alpha,\beta)\in\RR\times\RRp,
    \end{align}
    because $\nabla_\iota j(\bzero)=\bzero$ by assumption.
    This yields property \eqref{eq:dissipative_visc} by substituting the expressions in \eqref{eq:relationship_j_viscosity_2} into \eqref{eq:dissipative_visc_j}.
\end{proof}

\subsection{Specialized fluid models}

The numerical examples studied below in \cref{sec:numerics} include an incompressible fluid and two sea ice models. These models are based on slight modifications of \cref{eq:navier_stokes}, which we present next.

\subsubsection{Incompressible fluids}\label{subsubsec:incompressible}

A fluid is said to be incompressible if it can only experience isochoric motions; that is, the material volume of any part of the fluid remains constant in time. This condition reduces the equation of conservation of mass \eqref{eq:conservation_mass} to the divergence-free condition on the velocity field: $\nabla\cdot\bu = 0$ in $\Omega$, see \cite[section 1.4]{truesdell2000}. As a result, $\alpha = \tr{\bD\bu} = 0$ for all velocity fields, and the bulk viscosity $\zeta$ is not well-defined. Thus, the local dissipation potential depends only on the second invariant $\beta$, i.e.~$j = j(\beta)$, and the dissipation potential is given by
\begin{align}
    \Jc_{\mathrm{inc}}(\bu):=\int_\Omega j(\beta(\bu))\,\dd x.
\end{align}
Moreover, for incompressible flows, an equation of state for the pressure is no longer required. In fact, the pressure $p$ plays the role of a Lagrange multiplier that enforces the incompressibility constraint.

\subsubsection{Depth-integrated sea ice fluid models}\label{subsubsec:seaice}
Sea ice models are generally derived by integrating the three-dimensional momentum balance equation in the vertical direction along the thickness of the ice cover \cite{gray1994}. Denoting the ice-thickness field with $H:\Omega\to\RRp$ and the concentration field with $A:\Omega\to[0,1]$, the thickness-integrated momentum balance equation over an ice-covered domain $\Omega\subset\RR^2$ is given by
\begin{align}\label{eq:sea_ice_momentum_balance}
    \rho_i H\left(\frac{\partial \bu}{\partial t} + \left( \bu\cdot\nabla\right)\bu\right) -  \nabla\cdot \bsigma = \btau_o(\bu) + \btau_a \quad \text{in $\Omega$}.
\end{align}
In \eqref{eq:sea_ice_momentum_balance}, $\rho_i$ is the density of ice, and $\btau_o(\bu)$ and $\btau_a$ denote the ocean and atmospheric drag, respectively. These are parameterized in terms of drag coefficients $C_o$, $C_a$, densities $\rho_o$, $\rho_a$, and ocean and atmosphere velocity fields $\bu_o$, $\bu_a$ as
\begin{align}\label{eq:drag_forces}
    \btau_o(\bu) := \rho_oC_o|\bu_o-\bu|(\bu_o-\bu) \quad\text{and} \quad \btau_a := \rho_aC_a|\bu_a|\bu_a.
\end{align}
The ocean and wind drag correspond with the Fr\'echet derivative of the potential
\begin{align}\label{eq:sea_ice_force_potential}
    \Phi(\bu):= -\frac{1}{3}\int_\Omega\rho_oC_o|\bu_o-\bu|^3\,\dd x + \int_\Omega\rho_aC_a|\bu_a|\,(\bu_a\cdot\bu)\,\dd x.
\end{align}
Thus, if the sea ice rheology can be written in terms of a dissipation potential $\Jc(\bu)$ and we neglect the inertial terms in \eqref{eq:sea_ice_momentum_balance}, the momentum balance equations correspond with the optimality conditions of the energy potential $\Fc$ defined in \eqref{eq:condition_on_F}.

Additional advection equations must be solved for the thickness and concentration fields. For simplicity, the sea ice problems in \cref{subsec:vp} and \cref{subsec:dem} set the ice thickness constant. If we neglect thermodynamic growth, the concentration field evolves according to
\begin{align}\label{eq:advection_A}
    \frac{\partial A}{\partial t} + \nabla \cdot \left( \bu A\right) = 0,
\end{align}
an equation which can be derived from \eqref{eq:conservation_mass} by introducing $\rho = \rho_i A$.

\section{Inferring the dissipation potential from velocity data}\label{sec:computational_framework}

Computationally, we infer the local dissipation potential $j$ by representing it in terms of a parametric family of continuously differentiable functions $j_\theta$ with parameters $\theta\in\Theta$. As data, we take $N$ snapshots  of velocity fields $\bu_{\mathrm{data},1}, \ldots, \bu_{\mathrm{data},N}$, with $N$ typically small or even 1, as in \cref{subsec:ns,subsec:hb}. We solve the following optimization problem:
\begin{align}\label{eq:minimization}
    \theta^\ast\in\argmin_{\theta\in\Theta} \Lc(\theta), \quad \Lc(\theta) := \frac{1}{2} \sum_{i=1}^N \lambda_i\| \bu_i(\theta) - \bu_{\mathrm{data},i}\|^2 + \Pi(\theta).
\end{align}
In \eqref{eq:minimization}, $\bu_i(\theta)$ is the velocity field at snapshot $i$ obtained by solving \cref{eq:navier_stokes}, the equations of fluid motion, with the constitutive law defined by $j_\theta$. The positive scalars $\lambda_i>0$ are weights associated to each velocity snapshot, $\|\cdot\|$ denotes the $\ell^2$ norm evaluated at a given set of observation points, and $\Pi$ is the penalization term defined in the next section, in \eqref{eq:penalty}. Problem \eqref{eq:minimization} is a PDE-constrained optimization problem, with the constraint given by the fluid model implicitly defined by $\bu(\theta)$. In \cref{subsec:numerical_implementation} we explain that, in practice, we solve this optimization problem with gradient-based methods.

\subsection{A structure-preserving neural network architecture (ICNNE)}\label{subsec:structure_preserving_architecture}

\cref{thm:properties_j} provides sufficient conditions on the dissipation potential $j_\theta$ to guarantee that the inferred constitutive law is frame-indifferent, dissipative, and derived from a convex dissipation potential $\Jc$. Frame-indifference follows from representing $j_\theta$ in terms of the invariants $\alpha$ and $\beta$. We enforce the convexity of $j_\theta$ with input-convex NNs \cite{amos2017} with softplus activation functions. To guarantee that $j_\theta$ is increasing in $\beta$, we enforce evenness in $\beta$ by writing
\begin{align}
    j_\theta(\alpha,\beta) = \frac{1}{2} \left( g_\theta(\alpha,\beta) + g_\theta(\alpha,-\beta)\right) \quad \forall(\alpha,\beta)\in\RR\times\RRp,
\end{align}
where $g_\theta$ is an input-convex neural network defined on $\RR^2$. That is, $j_\theta$ is the average of $g_\theta$ and its reflection with respect to the $\alpha$-axis. Since we exactly enforce evenness in $\beta$, we call this NN architecture the ICNNE model. As a result, we have that
\begin{align}\label{eq:null_derivative_evenness}
    \frac{\partial j_\theta}{\partial \beta}(\alpha,0)=0\quad\forall\alpha\in\RR.
\end{align}
Since $j_\theta$ is convex in $\beta$, we have that $\partial j_\theta/\partial\beta$ is increasing in $\beta$. Together with \eqref{eq:null_derivative_evenness}, it follows that $\partial j_\theta/\partial\beta(\alpha,\beta)\geq 0$ for all $(\alpha,\beta)\in\RR\times\RRp$, and therefore $j_\theta$ is increasing in $\beta$. Thus, if $\nabla_\iota j_\theta(0,0) = \bzero$, the three properties listed in \cref{thm:properties_j} are satisfied for $j_\theta$. We enforce this last property by penalizing the $\alpha$-derivative of $j_\theta$, such that
\begin{align}\label{eq:penalty}
    \Pi(\theta) := \lambda_p |\partial_\alpha j_\theta(0,0)|^2.
\end{align}

From a numerical perspective, \eqref{eq:null_derivative_evenness} is beneficial because it alleviates potential numerical instabilities that arise from the factor $1/\beta$ in \eqref{eq:viscous_stresses_j}. Moreover, \eqref{eq:null_derivative_evenness} is a natural condition to enforce whenever the effective shear viscosity $\eta$ is bounded near $\beta=0$. Under this condition, we have that $\lim_{\beta\to 0}\eta(\alpha,\beta)\beta = 0$ and therefore \eqref{eq:null_derivative_evenness} follows from \eqref{eq:relationship_j_viscosity_2} because $j_\theta$ is continuously differentiable.

An ICNNE-based dissipation potential has linear growth for large values of $(\alpha, \beta)$ because we use softplus activation functions and input-convex NNs. Therefore, the extension of \cref{thm:well-posed} proposed in Remark \ref{remark:extension-well-posedness} would establish the existence of solutions for ICNNE-based models.

\subsection{Alternative NN architectures}

Besides the ICNNE architecture described in \cref{subsec:structure_preserving_architecture}, we also explore alternative NN architectures and problem formulations.
To highlight the advantages of the ICNNE architecture, we use three different NN architectures: the input-convex NN architecture without evenness in $\beta$ and standard feed-forward NNs with and without evenness. We refer to these three models as ``ICNN'', ``SNNE'', and ``SNN'', respectively. They are listed as models to compare the ICNNE method with in \cref{tab:architectures}, together with the model SNNsq described below.

\begin{table}[t]
    \centering
    \caption{The five models explored in the numerical examples of \cref{sec:numerics}. The ICNNE architecture, with the penalization \eqref{eq:penalty}, enforces all of the desired properties that ensure the model is well-posed.}
    \label{tab:architectures}
    \begin{tabular}{l c c c c c}
    \toprule
    \rule{0pt}{2.5ex}
    & \textbf{ICNNE} & ICNN & SNN & SNNE & SNNsq \\
    \midrule
    Convex & $\checkmark$ & $\checkmark$ & $\times$ & $\times$ & $\times$ \\
    Even $j(\cdot, -\beta)=j(\cdot, \beta)$ & $\checkmark$ & $\times$ & $\times$ & $\checkmark$ & $\times$ \\
    $\beta $ & $|\bS\bu|$ & $|\bS\bu|$ & $|\bS\bu|$ & $|\bS\bu|$ & $|\bS\bu|^2$\\
    \bottomrule
    \end{tabular}
\end{table}

\subsection{Alternative formulation with squared second invariant}\label{subsubsec:beta_squared_form}

In addition to different NN architectures, we also compute numerical results with a different formulation of the viscous stress tensor \eqref{eq:viscous_stresses_j}. In \cref{subsec:structure_preserving_architecture}, we explain that the factor $1/\beta$ in \eqref{eq:viscous_stresses_j} can lead to numerical instabilities. To reduce these instabilities, we consider an alternative formulation, defining $\beta = |\bS\bu|^2$ instead of $\beta = |\bS\bu|$. The viscous stress tensor can then be written as
\begin{align}\label{eq:viscous_stresses_j_alternative}
    \Cc(\bD\bu) = \frac{\partial j}{\partial \alpha}(\alpha,\beta)\, \bI + \frac{\partial j}{\partial \beta} (\alpha,\beta)\,\bS\bu
\end{align}
and therefore $2\eta = \partial j/\partial \beta$. Since the factor $1/\beta$ is no longer present, we do not need to enforce the evenness condition on $j$ to ensure numerical stability. Although this is a more straightforward formulation, enforcing convexity on $j$ makes the potential grow at least linearly in $\beta=|\bS\bu|^2$. This restricts the constitutive law
to be Newtonian or shear-thickening for the following reason: If $j$ is twice differentiable and convex in $\beta = |\bS\bu|^2$, we must have that $\partial\eta/\partial\beta \geq 0$ because $\partial^2 j/\partial \beta^2 \geq 0$. Many fluids of practical interest are not shear-thickening, see the examples in \cref{subsec:hb,subsec:vp}. For this reason, we only explore this alternative formulation with the SNN architecture, which does not enforce convexity on $j$. Throughout this paper, we refer to this model as ``SNNsq'', see \cref{tab:architectures}.

\subsection{Numerical implementation combining PyTorch and Firedrake} \label{subsec:numerical_implementation}

We implement the framework for inferring the potential $j_\theta$ described in \cref{sec:computational_framework} by integrating two Python modules: the ML library PyTorch \cite{paszke2019} and the FE library Firedrake \cite{ham2023}. PyTorch is employed to construct the NN-based dissipation potential $j_\theta$ and to minimize the loss function $\Lc(\theta)$ using its implementation of the limited-memory Broyden–Fletcher–Goldfarb–Shanno (LBFGS) algorithm \cite{wright1999}. Because $\Lc(\theta)$ is defined in terms of the velocity misfit, its evaluation necessitates solving the PDEs implicitly defined by $\bu = \bu(\theta)$ that describe the fluid motion. We compute gradients of $\Lc(\theta)$ via the adjoint method, so we also require the solution of the associated adjoint problems. Both the PDEs and their adjoint systems are solved in Firedrake using the finite element method, a provably-convergent and robust approach to solving PDEs numerically. Firedrake and PyTorch are coupled by evaluating the NN-based potential $j_\theta$ and its derivatives in PyTorch at the mesh's quadrature points. To avoid dividing by zero, we use the following regularization of the second invariant in our implementation:
\begin{align}
    \beta_\epsilon := \left(\frac{1}{2}\, \bS\bu:\bS\bu + \epsilon^2\right)^{1/2}
\end{align}
for a parameter $\epsilon > 0$ whose value is specified for each numerical test in \cref{sec:numerics} below. Thus, in practice, we define our PDEs in terms of the local parameterized potential $j_\theta = j_\theta(\alpha,\beta_\epsilon)$.

We leverage the automatic differentiation capabilities in Firedrake and PyTorch to avoid explicit manual derivative calculations when applying Newton’s method and constructing the adjoint system. Firedrake and PyTorch are coupled in such a way that gradients can be computed efficiently and in a way that does not depend on the number of NN parameters. Indeed, the main computational cost in both loss evaluation and gradient calculation lies in solving the governing PDE for fixed NN parameters; the adjoint PDE is linear and can therefore be solved at a cost comparable to a single Newton step for the forward PDE. In contrast to these PDE operations, the time spent on NN evaluation and backpropagation is negligible. Since we employ a deterministic quasi-Newton method for loss minimization, we usually achieve satisfactory convergence within about 100–300 iterations (see, for example, \cref{fig:ns_results}). Consequently, the learning framework naturally extends to the case of 3D velocity data: as long as the governing PDE can be solved in Firedrake, generalizing the method to learn from 3D velocity data is straightforward.
The numerical results can be reproduced with the repository accessible with the following link: \url{https://bitbucket.org/gonzalogddiego/convfit/src/v1.0.0/}.

\section{Numerical Results}\label{sec:numerics}
We apply the computational framework described in \cref{sec:computational_framework} to infer four different constitutive laws, which are summarized in \cref{tab:numerical_exps}, with the five different models described in \cref{tab:architectures}.  For compressible flow problems, we only consider variants of our ICNNE method because, to our knowledge, existing methods to recover constitutive laws for fluids from velocity data are strictly limited to incompressible fluids \cite{xu2021,sunol2025} or compressible problems in one dimension \cite{dediego2026}.
For the HB model, which is incompressible,
we compare our method with that of \cite{sunol2025} to demonstrate our approach's accuracy in the context of incompressible fluid problems. Below, we summarize the setup of our numerical examples.

\begin{table}
\caption{Overview of numerical experiments: Compressible Navier--Stokes (NS), incompressible Herschel--Bulkley (HB), compressible viscous plastic (VP), and discrete element method (DEM) data.
The table indicates whether the model is compressible or incompressible, whether it includes inertia, and the target potential and unknowns. It also lists the number of velocity snapshots and points per snapshot used for learning, and the noise level.
}\label{tab:numerical_exps}
\centering
\begin{tabular}{c|cccccc}
\toprule
 model & \parbox{10ex}{\centering compress-\\ible?} & inertia? & target & \parbox{7ex}{\centering\#snap-\\shots} & \parbox{7ex}{\centering\#data\\ pts.} & \parbox{6ex}{\centering noise} \\ \midrule
NS & yes & yes & $j(\alpha, \beta)$ & 1 & 500 & 1\% \\
HB & no & no & $j(\beta)$ & 1 & 1500 & 10\% \\
VP & yes & no & $j(\alpha,\beta,A)$, $p^\ast$ & 8 & 1000 & 5\% \\
DEM & yes & no & $j(\alpha,\beta,A)$, $p^\ast$ & 3 & 100 & --- \\
\bottomrule
\end{tabular}
\end{table}

\paragraph{NN architectures and parameterization}
For each model, we consider four architectures, denoted by their layer widths: (20), (5, 5), (10, 10), and (5, 5, 5). These deliberately small networks suffice for our test problems, consistent with prior work on input-convex neural networks for constitutive modeling \cite{geuken2025}.
For the NS and HB examples, we use standard multilayer perceptrons. Since
in the VP and DEM sea-ice models, the concentration $A$ strongly affects the viscosity (see, e.g., the exponential term in the pressure law \eqref{eq:pressure_vp}),
we use the NN parameterizations
\begin{align}
    j^{\mathrm{scaled}}_\theta(\alpha, \beta, A) = \exp(\xi_\theta(A))\, j_\theta(\alpha,\beta,A).
\end{align}
Here, $j_\theta$ corresponds with one of the NN models described in \cref{sec:computational_framework} (e.g.~ICNNE (5 5)). The function $\xi_\theta:[0,1]\to\RR$ is a standard multilayer perceptron with the same width and number of layers as $j_\theta$.
To gather statistically significant results for each model and architecture, we repeat the training over 10 different initializations. Therefore, each model is trained a total of 40 times.

\paragraph{Loss function}
In our experiments, we infer the network parameters (and in the VP and DEM datasets also an additional parameter $p^\ast>0$) by minimizing the loss $\Lc(\theta)$ defined by \eqref{eq:minimization} with penalty \eqref{eq:penalty}, i.e.,
\begin{align}\label{eq:loss}
    \Lc(\theta) := \frac{1}{2} \sum_{i=1}^N \lambda_i\| \bu_i(\theta) - \bu_{\mathrm{data},i}\|^2 + \frac{\lambda_p}{2} |\nabla_\iota j_\theta(0,0)|^2.
\end{align}
Here, the norm is evaluated at the observation points $\bx_1$,\ldots,$\bx_{N_o}$, that is, $\| \bv\|^2:=\sum_{k=1}^{N_o}|\bv(\bx_k)|^2$.
We set $\lambda_i := \|\bu_{\mathrm{data},i}\|_{L^\infty(\Omega)}^{-2}$ to ensure consistency across velocity scales of the different snapshots. The penalization parameter $\lambda_p$ is
specified separately for each problem. The observation points $\bx_1,\ldots,\bx_{N_o}$ are randomly sampled from a uniform distribution over the domains. The number of snapshots $N$ and observation points $N_o$ are specified in \cref{tab:numerical_exps}.

\paragraph{Data generation}
In the NS, HB and VP problems we learn a known constitutive law from noisy velocity data that we generate by solving the underlying PDEs with known constitutive law. Then, we add additive noise to the uncorrupted velocity observations $\hat{\bu}_{\mathrm{data},i}$, prescribing the unscaled noise level $\lambda_n$ as follows:
\begin{align}
    \bu_{\mathrm{data},i}(\bx) = \hat{\bu}_{\mathrm{data},i}(\bx) +\lambda_{n}\|\hat{\bu}_{\mathrm{data},i}\|_{L^\infty(\Omega)}\varepsilon_n, \quad \varepsilon_n\sim\mathcal{N}(0,\bI).
\end{align}
For the DEM example, we use data generated with a discrete element method (DEM) for which the underlying constitutive law is unknown.

\paragraph{Optimization}
We use PyTorch's implementation of the L-BFGS algorithm, which allows one to specify a number of inner and outer iterations, which we set to 10 and 30, respectively. Thus, for each problem, we perform around 300 evaluations of the loss functional. The number of evaluations may be lower if convergence occurs early ($\max|\nabla\Lc|<10^{-7}$ or the absolute change in $\Lc$ falls below $10^{-9}$), or if the PDE solver fails for 10 consecutive evaluations. We do not use NN weight regularization because it did not improve results, and we prefer formulations with no or few tunable parameters.

We pre-train our NN with quadratic (i.e., Newtonian) dissipation potential data. We limit this pre-training to a few steps because, when working with input-convex NNs, we sometimes find that initializing our training with a model that strongly resembles a quadratic functional may stall convergence; input-convex NNs are known to be difficult to train due to the positivity requirements on their weights \cite{bunne2022, uscidda2023}. We consider separate training and test datasets.

\paragraph{Performance diagnostics}
The training dataset is used to learn the dissipation potential and the test dataset is used to evaluate the performance of the learned constitutive law.
As performance measure for these test problems, we display the relative $L^2$ error of the velocity field, defined by $\|\bu(\theta) - \bu_{\mathrm{exact}}\|_{L^2(\Omega)}/\|\bu_{\mathrm{exact}}\|_{L^2(\Omega)}$, where $\bu_{\mathrm{exact}}$ is the true velocity field of the test dataset.
These results are shown in \crefrange{tab:ns_test_error}{tab:vp_test}, with \textit{N.~fail} indicating that Newton's method fails when training the model or when attempting to solve the test problem, over all initializations. This is the typical failure mode we observe in our computations, and it often arises during the PDE-based training. The failure of the PDE solver is intimately related to the main issue we address with the proposed ICNNE model: if the NN-based model for the constitutive law does not guarantee the well-posedness of the PDE, training may be stalled because the NN-based PDE becomes unsolvable and we cannot evaluate the loss function $\mathcal{L}(\theta)$. For the DEM problem, since we do not have a test dataset, we only evaluate the performance of the learned constitutive law with the optimized loss function value.

For the NS, HB and VP problems, where the true constitutive relation is known,
we present the best test error, the corresponding loss value and architecture.
In addition, we report the number of training runs (out of 40) for which Newton's method breaks down when solving the test problem, together with the average number of Newton iterations for runs that succeed in solving the test problem.

\subsection{Compressible Navier--Stokes equations}\label{subsec:ns}

As a first example, we infer the compressible Navier--Stokes (NS) equations with constant bulk and shear viscosities $\zeta_0$ and $\eta_0$. The dissipation potential is then the quadratic function
\begin{align}\label{eq:newtonian_quadratic_potential}
    j(\alpha,\beta) = \frac{1}{2}\zeta_0 \alpha^2 + 2\eta_0 \beta^2,
\end{align}
and we generate data with $\zeta_0=10$ and $\eta_0 = 0.1$. Since the fluid is compressible, we need an equation of state to compute the pressure $p$ and density $\rho$, which we take to be $p = C \rho^{\gamma}$ with $C = 1$ and $\gamma=1.4$. The momentum equation \eqref{eq:momentum_balance} includes the inertial terms and is discretized with continuous piecewise linear velocity elements (CG1). The density and pressure are represented with piecewise constant elements (DG0) and obtained by solving the conservation of mass equation \eqref{eq:conservation_mass} from the initial state $\rho(\bx,t=0)=1$ to a stationary state using standard upwinding. During PDE-constrained training, the stationary density and pressure fields are treated as fixed input data, and only the momentum equation is solved.

\paragraph{Training problem (curved channel)}
The training data is generated by solving for a stationary flow through the curved channel depicted in the top left of \cref{fig:ns_results}. The  corresponding stationary density field is shown in the bottom left panel of the same figure.
The domain is periodic in the horizontal direction with no-slip boundary conditions $\bu = \bzero$ along the upper and lower walls, and the flow is driven by the body force $\bff = (1,0)^\top$. The domain is meshed with 2000 cells. The models are trained on a single snapshot of velocity data corrupted with 1\% noise ($\lambda_n=0.01$) and sampled at $N_o = 500$ uniformly distributed points. The penalization parameter is set to $\lambda_p = 10^{4}$.

\begin{figure}[tb]
    \centering
    \includegraphics[width=0.98\textwidth]{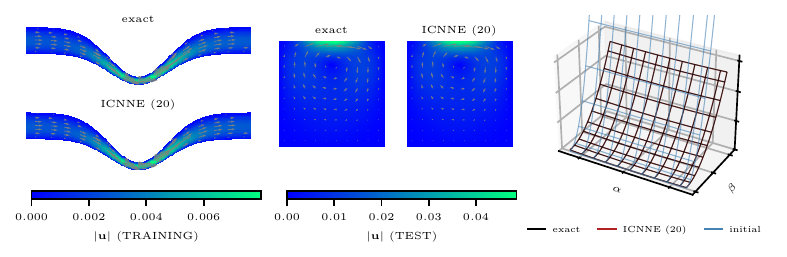}\\
    \includegraphics[width=\linewidth]{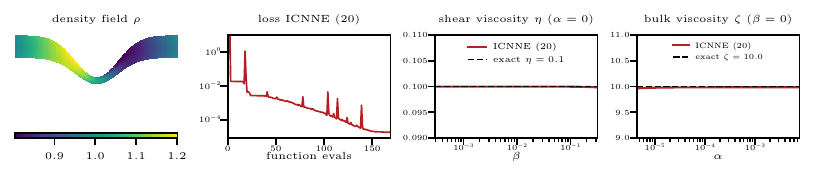}\\[-2ex]
    \caption{Compressible NS:
    Top row: The left and middle columns show the velocity field for the training and test problems, respectively, while the right panel compares the learned dissipation potential with the exact potential and the initial guess.
    Bottom row: Stationary density field of the training problem (left), values of the loss function $\Lc$ across the L-BFGS function evaluations (second panel), and effective shear and bulk viscosities at $\alpha = 0$ and $\beta = 0$, respectively (right two panels), computed with the best scoring model, ICNNE (20). }
    \label{fig:ns_results}
\end{figure}

\paragraph{Test problem (lid-driven cavity)}
The trained models are evaluated on a regularized lid-driven cavity problem on the unit square $\Omega = [0,1]^2$ (middle column of \cref{fig:ns_results}). No-slip conditions $\bu=\bzero$ are enforced on the left, right, and bottom boundaries, while the top boundary moves with the smooth profile
\begin{align}
    \bu(x,y=1) = \left(\begin{array}{c}
         \frac{1}{2} u_{\mathrm{top}}\left(1-\cos{(2\pi x)} \right)\\
        0
    \end{array} \right)
\end{align}
with $u_{\mathrm{top}}=0.05$. The test mesh has 5000 cells. As for training, the ground truth density and pressure are first obtained from a stationary solve and then held fixed when solving the momentum equation with the learned constitutive laws.

\paragraph{Results}
The optimal loss values and relative $L^2$ velocity errors for the test problem over all architectures and initializations are shown in \cref{tab:ns_test_error}. The best results are obtained with the ICNNE model with network (20). This model is used to compute the velocity fields shown in the top row of \cref{fig:ns_results}, whose right panel compares the learned potential with the exact potential and the initial guess. The bottom row of this figure shows the decay of the loss over the L-BFGS function evaluations and the learned effective shear and bulk viscosities along $\alpha=0$ and $\beta=0$, respectively. In \cref{tab:ns_test_error}, for the ICNN and SNN models, we observe that, despite the training being stalled due to Newton's method failing, the models are still capable of solving the test problem successfully, although a very large relative error is achieved in these cases.

\begin{table}[tbh]
    \centering
    \caption{Compressible NS: Optimal values of loss function $\Lc(\theta)$ and relative $L^2$ error for test problem (lid-driven cavity flow) achieved over the 4 architectures, each trained over 10 different initializations of the network parameters $\theta$. For each model, we show the best test error achieved, together with the corresponding loss value and the network architecture which achieves the best test error. The final two rows show the number of runs for which Newton's method fails when solving the test PDE problem, and the average number of PDE Newton iterations, computed over the runs that solve the test problem successfully.}
    \label{tab:ns_test_error}
        \small
    \begin{tabular}{l c c c c c }
    \toprule
     & \textbf{ICNNE} & ICNN & SNN & SNNE & SNNsq \\
    \midrule
    Loss & \textbf{\num{1.9e-05}} & \textit{N. fail} & \textit{N. fail} & \num{1.9e-05} & \num{2.2e-05} \\
    Test error & \textbf{\num{1.9e-04}} & \num{6.8e-01} & \num{6.8e-01} & \num{4.7e-04} & \num{3.3e-03} \\
    Best network & \textbf{20} & 10 10 & 10 10 & 20 & 5 5 \\
    Fails (of 40) & \textbf{0} & 1 & 23 & 15 & 19 \\
    Avg. Newton & \textbf{2.5} & 133.7 & 82.1 & 5.9 & 5.0 \\
    \bottomrule
    \end{tabular}
\end{table}

\subsection{Herschel–Bulkley for an incompressible fluid}\label{subsec:hb}

The Herschel–Bulkley (HB) model is widely used for complex fluids that exhibit a yield stress. It is an incompressible fluid model, and therefore the corresponding equations of motion are those described in \cref{subsubsec:incompressible}. In its regularized form, the local dissipation potential is given by
\begin{align}\label{eq:hb_potential}
    j(\beta) = \begin{cases}
        2\eta_0 \beta^2 & \text{if } \beta \leq \beta_0, \\
        4\tau_0 \beta + \frac{4k}{n+1} \beta^{n+1} + C & \text{if } \beta > \beta_0,
    \end{cases}
\end{align}
where $\tau_0$ is the yield stress, $k$ is the consistency index, $n$ is the flow index, $\eta_0$ the limiting viscosity, and $C$ is a constant that ensures continuity of the dissipation potential. The corresponding effective shear viscosity is
\begin{align}\label{eq:hb_effective_visc}
    \eta(\beta) = \begin{cases}
        \eta_0 & \text{if } \beta \leq \beta_0, \\
        \frac{\tau_0}{\beta} + k \beta^{n-1} & \text{if } \beta > \beta_0.
    \end{cases}
\end{align}
For values of $n$ in $(0,1)$, the model exhibits shear-thinning behavior, for $n>1$, it exhibits shear-thickening behavior. We generate data with $\beta_0=0.001$, $k=0.5$, $n=0.5$, and $\tau_0=0.1$, so that the viscosity is shear-thinning, and set $\eta_0 = \tau_0/\beta_0 + k\beta_0^{n-1}$ to ensure continuity at $\beta = \beta_0$. We solve the incompressible Stokes equations using Taylor--Hood finite elements,
i.e., continuous quadratic elements (CG2) for the velocity and continuous linear elements (CG1) for the pressure, and use $\epsilon = 10^{-8}$ in the regularized invariant $\beta_\epsilon$.

\paragraph{Training problem (constricted channel)}
The training data is generated by solving for a stationary flow through the constricted channel shown in the left panel of \cref{fig:hb_results}, where the fluid is driven by a pressure gradient, which we set equal to 5. The domain is of unit length and its maximum height is 0.5. The constriction around $x = 0.5$ is modeled with two Gaussian curves on the upper and lower walls. The flow is periodic in the horizontal direction with no-slip boundary conditions $\bu = \bzero$ along the walls. We mesh the domain with 1250 quadrilateral cells and train the models on a single snapshot of velocity data corrupted with 10\% noise ($\lambda_n = 0.10$) and sampled at $N_o = 1500$ uniformly distributed observation points.

\paragraph{Test problem (3D toothpaste extrusion)}
The trained models are evaluated on a three-dimensional extrusion problem that mimics squeezing toothpaste from a tube (middle panel of \cref{fig:hb_results}). No-slip conditions are enforced on the tube walls with the exception of the side wall, where a Gaussian velocity profile of amplitude $u_{\max} = 10^{-3}$ is prescribed to drive the flow, and the nozzle end, which is a free-outflow boundary. The same Taylor--Hood discretization is used, this time on a three-dimensional tetrahedral mesh generated with gmsh.

\begin{figure}[bt]
    \centering
    \includegraphics[width=\textwidth]{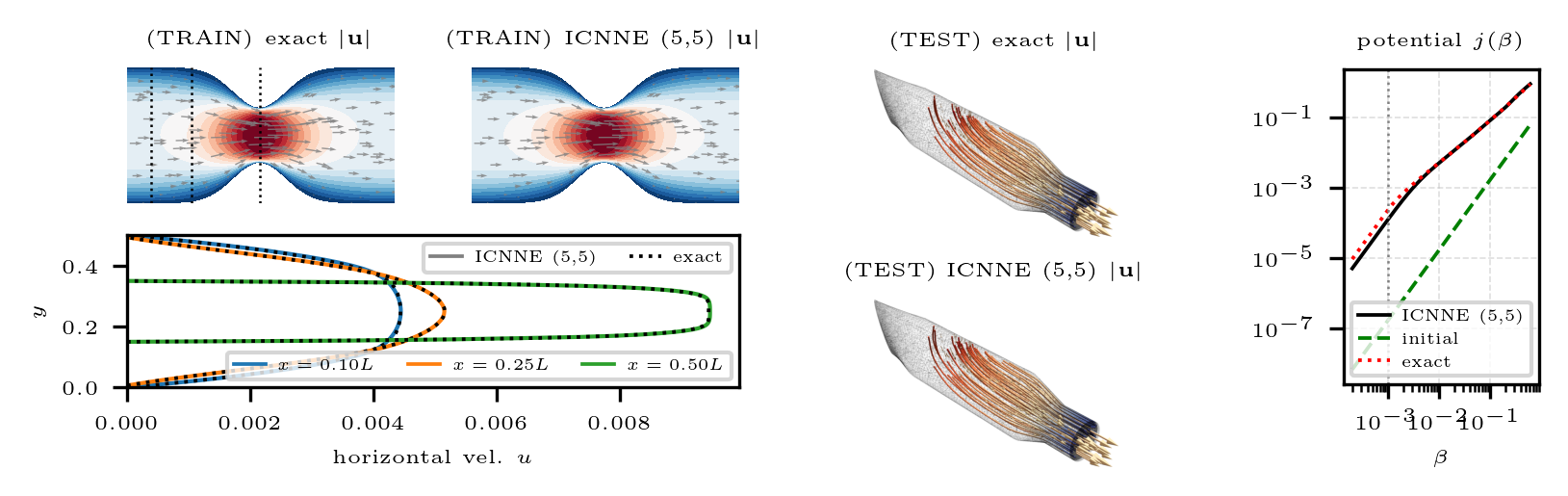}
    \caption{Incompressible HB Stokes:
    The left and middle panels show the exact and learned velocity fields for the training and test problems, respectively. In the bottom-left, we display the velocity cross-section profiles at $x=0.1L$, $0.25L$ and $0.50L$, comparing exact vs the model. The right panel shows the dissipation potential for the best scoring model.}
    \label{fig:hb_results}
\end{figure}

\paragraph{Comparison with the method of \cite{sunol2025}}
For this problem, we include an external comparison with the method from \cite{sunol2025} for inferring non-Newtonian models for incompressible fluids. Rather than a dissipation potential, this method parametrizes the effective shear viscosity directly as a monotone interpolation between a zero-shear viscosity $\eta_0$ and an infinite-shear viscosity $\eta_\infty$,
\begin{align}\label{eq:sunol_viscosity}
    \eta(\beta) = \eta_\infty\left[1 + \left(\frac{\eta_0}{\eta_\infty} - 1\right)\big(1 - F(z)\big)\right], \quad F(z) = \sum_{m=1}^{M} a_m\, \sigma\!\left(\frac{z-\mu_m}{s_m}\right),
\end{align}
where $z = \log(2\beta)$, $\sigma$ is the logistic sigmoid and $M=12$. The mixture parameters $(a_m,\mu_m,s_m)_{m=1}^M$ are the outputs of an NN with one hidden layer of 16 units that takes $z$ as input, with $a_m$ normalized by a softmax and $s_m>0$ enforced by a softplus, so that $F$ increases monotonically from 0 to 1; $\eta_0,\eta_\infty>0$ are also trainable, resulting in 646 parameters. Thus, $\eta$ is monotone in $\beta$, but, in contrast to our ICNNE model, the resulting stress is not guaranteed to derive from a convex dissipation potential. We implement \eqref{eq:sunol_viscosity} within our PyTorch--Firedrake framework (\cref{subsec:numerical_implementation}) and train it on the same data, with the same loss function, and with the same optimizer as our models. As for our models, the network is pre-trained to a Newtonian (constant-viscosity) model. Since the architecture is fixed by \cite{sunol2025}, we report results over 10 random initializations rather than the 40 runs used for our models.

\paragraph{Results}
The optimal loss values and relative $L^2$ velocity errors for the extrusion test problem are shown in \cref{tab:hb_test_error}. The best results are obtained with the ICNNE model with network (5 5); the right panel of \cref{fig:hb_results} shows the dissipation potential of this model. Although the method of \cite{sunol2025} achieves the same loss value during training as our best ICNNE model, it produces a larger test error than the ICNNE model. Similarly to the NS problem, we observe in \cref{tab:hb_test_error} that, for the SNNE and SNNsq models, the best test error is achieved by a model that fails to complete the PDE-based training, yet successfully solves the test problem. In the case of the SNNE and SNNsq models, the LBFGS method during training broke down early on, such that the final fluid model is very close to the Newtonian model used in the initialization. This explains the very low number of iterations required for solving the test problems for SNNE and SNNsq in table \ref{tab:hb_test_error}.

\begin{table}[t]
    \centering
    \caption{Incompressible HB Stokes: Optimal values of loss function and relative $L^2$ error for test problem (3D toothpaste extrusion), together with fails and average number of iterations with Newton's method. See caption for \cref{tab:ns_test_error} for details. For the HB problem, we include an external comparison with the method from \cite{sunol2025} for incompressible flow problems, which we run for a fixed architecture over 10 seeds.}
    \label{tab:hb_test_error}
    \small\setlength{\tabcolsep}{4pt}
    \begin{tabular}{l c c c c c c }
    \toprule
     & \textbf{ICNNE} & ICNN & SNN & SNNE & SNNsq & \cite{sunol2025} \\
    \midrule
    Loss & \textbf{\num{8.8e-03}} & \textit{N. fail} & \textit{N. fail} & \textit{N. fail} & \textit{N. fail} & \num{8.8e-03} \\
    Test error & \textbf{\num{1.7e-02}} & \textit{N. fail} & \textit{N. fail} & \num{2.1e-01} & \num{2.1e-01} & \num{3.7e-02} \\
    Best network & \textbf{5 5} & -- & -- & 5 5 & 5 5 5 & 10 10 \\
    Fails (of 40) & \textbf{0} & 40 & 40 & 0 & 0 & 1/10 \\
    Avg. Newton & \textbf{3.5} & -- & -- & 1.4 & 1.1 & 11.8 \\
    \bottomrule
    \end{tabular}
\end{table}

\subsection{The compressible viscous-plastic model for sea ice}\label{subsec:vp}

The viscous-plastic model for sea ice, also known as Hibler's model \cite{hibler1979}, is the most common model for sea ice. It is a two-dimensional compressible fluid model that depends on the sea ice concentration $A:\Omega\to[0,1]$ and thickness $H:\Omega \to \RRp$ fields, and whose equations of motion are the thickness-averaged equations presented in \cref{subsubsec:seaice}.
Here, we assume $H = H_0$ is fixed and we neglect inertial terms in the momentum balance equation \eqref{eq:sea_ice_momentum_balance}. We infer the concentration-dependent constitutive law, which is defined in terms of the potential:
\begin{align}\label{eq:vp_potential}
    j(\alpha,\beta, A) = p(A) \sqrt{4e^{-2}\beta^2 + \alpha^2 + \Delta_{\min}^2}.
\end{align}
In \eqref{eq:vp_potential}, the pressure law is given by
\begin{align}\label{eq:pressure_vp}
    p(A) = p^\ast H_0\exp{(-C(1-A))},
\end{align}
where $p^\ast$ is the ice strength and $C$ an empirical constant. The parameter $e$ is the eccentricity of the plastic yield curve, and $\Delta_{\min}$ is a regularization parameter that yields a viscous law for low strain-rates. The potential \eqref{eq:vp_potential} yields the effective bulk and shear viscosities
\begin{align}
    \zeta(\alpha,\beta,A) = \frac{p(A)}{\sqrt{4e^{-2}\beta^2 + \alpha^2 + \Delta_{\min}^2}} \quad \text{and} \quad \eta(\alpha,\beta,A) = e^{-2}\zeta(\alpha,\beta,A).
\end{align}
We generate data with eccentricity $e = 2$, non-dimensional regularization parameter\footnote{Note that this value is larger than values typically adopted in VP-style sea ice models, where values as low as $\Delta_{\min} = \SI{2e-9}{\per\second}$
are used \cite{shih2023, mehlmann2021}. Values this small produce nearly discontinuous velocity features that can be viewed as sea ice opening, and they also introduce a strong nonlinearity into the governing PDE that may necessitate specialized solvers.} $\Delta_{\min} = \SI{1e-6}{\per\second}$, fixed ice thickness $H_0 = \SI{1}{\meter}$, concentration-decay coefficient $C = 20$, and ice strength $p^\ast=\SI{27.5e3}{\newton\per\meter\squared}$. During training, $C$ and $H_0$ are held fixed, while $p^\ast$ is treated as an additional trainable parameter that is jointly optimized with the network weights $\theta$. The drag terms \eqref{eq:drag_forces} use $C_o = 5.5\times10^{-3}$, $C_a = 1.2\times10^{-3}$, $\rho_o = \SI{1026}{\kilo\gram\per\meter\cubed}$, and $\rho_a = \SI{1.3}{\kilo\gram\per\meter\cubed}$.
We discretize the velocity with continuous piecewise bilinear elements on quadrilateral cells (CG1) and the concentration with piecewise constant elements (DG0), and use $\epsilon = 10^{-6}$ in the regularized invariant $\beta_\epsilon$. The equations of motion are non-dimensionalized with the characteristic length $[x] = \SI{1e5}{\meter}$, velocity $[u] = \SI{1}{\meter\per\second}$, and stress $[\sigma] = \rho_oC_o[u]^2[x]$. The penalization parameter is set to $\lambda_p = 10^{2}$.

\paragraph{Training problem (time-dependent cyclone)}
The training data is generated by solving the time-dependent compressible VP momentum equation coupled with the concentration advection equation \eqref{eq:advection_A} on a $512\,\mathrm{km}\times 512\,\mathrm{km}$ square domain with no-slip boundary conditions on all walls, see \cref{fig:vp_setup}. The domain is meshed with $100\times 100$ structured quadrilateral cells. The ice is driven primarily by a time-dependent wind field corresponding to a cyclone that travels from the center to a corner and back over a simulation of duration $T = 192\,\mathrm{h}$ (8 days), consistent with the benchmark of \cite{mehlmann2021}. The cyclone reverses direction at day 4, producing both cyclonic and anticyclonic forcing; its maximum wind speed is $12.9\,\mathrm{m\,s^{-1}}$. The ocean current is a steady circular vortex of amplitude $\sim 0.01\,\mathrm{m\,s^{-1}}$. The concentration is initialized at $A_0 = 0.95$ and advected each time step with an upwind scheme using a time step $\Delta t = 0.5\,\mathrm{h}$. We save $N = 8$ snapshots of the velocity and concentration at evenly spaced times over the simulation. The velocity data is corrupted with $5\%$ noise ($\lambda_n = 0.05$), and during the training, the loss is restricted to $N_o = 1000$ random observation points distributed uniformly over the domain.

\begin{figure}[tbh]
    \centering
    \includegraphics[width=\linewidth]{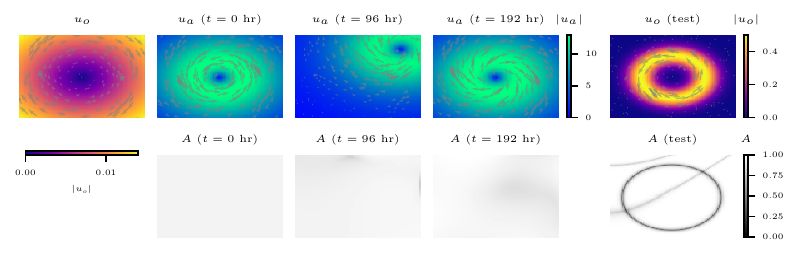}
    \caption{Compressible VP: ocean, wind, and concentration fields for the training and test problems. The first four columns to the left depict the training setup: the ocean velocity, which is stationary, together with the wind velocity and concentration at three time instants. The column to the right shows the ocean velocity and concentration field used for the test problem.}
    \label{fig:vp_setup}
\end{figure}

\paragraph{Test problem (steady vortex with heterogeneous concentration)}
The trained models are evaluated on a separate steady-state problem on a doubly-periodic $512\,\mathrm{km}\times 512\,\mathrm{km}$ square domain, meshed with $50\times 50$ cells, see the right column of \cref{fig:vp_setup}. The ocean current is a circular vortex centered at the origin, and the concentration field is spatially heterogeneous with low-concentration patches designed to produce large concentration gradients, following the benchmark of \cite{mehlmann2021}, which is designed to strain nonlinear solvers. A single velocity snapshot is generated with the exact VP law and is used to compute the relative $L^2$ error reported in \cref{tab:vp_test}.

\paragraph{Results}
The optimal loss values and relative $L^2$ velocity errors for the test problem are presented in \cref{tab:vp_test}. In \cref{fig:vp_results}, we also plot results and the dissipation potential of the model that achieves the lowest test error, an ICNNE model with a (5 5) architecture. The test case we consider is particularly challenging due to the strong discontinuities in the concentration field (see \cref{fig:vp_setup}). As a result, Newton's method is incapable of converging for practically all non-ICNNE models. The ICNNE model achieves an optimal test error which is almost one order of magnitude lower than the other models, and only one in four ICNNE models fail to converge when solving the test problem. The middle panel in \cref{fig:vp_results} shows that the learned dissipation potential closely matches the exact potential over the range of strain rates present in the training data.

\begin{table}[tbh]
    \centering
    \caption{Compressible VP: Optimal values of loss function and relative $L^2$ error for test problem (ocean vortex over heterogeneous concentration), together with fails and average number of iterations with Newton’s method. See \cref{tab:ns_test_error} for details.}
    \label{tab:vp_test}
    \begin{tabular}{l c c c c c }
    \toprule
     & \textbf{ICNNE} & ICNN & SNN & SNNE & SNNsq \\
    \midrule
    Loss & \textbf{\num{7.6e-04}} & \num{1.1e-03} & \textit{N. fail} & \num{5.7e-03} & \num{3.3e-02} \\
    Test error & \textbf{\num{7.4e-02}} & \num{1.2e-01} & \textit{N. fail} & \num{4.6e-01} & \num{4.3e-01} \\
    Best network & \textbf{5 5} & 5 5 5 & -- & 5 5 & 5 5 5 \\
    Fails (of 40) & \textbf{11} & 38 & 40 & 37 & 39 \\
    Avg. Newton & \textbf{14.0} & 69.5 & -- & 4.0 & 7.0 \\
    \bottomrule
    \end{tabular}
\end{table}

\begin{figure}[tbh]
    \centering
    \includegraphics[width=0.98\textwidth]{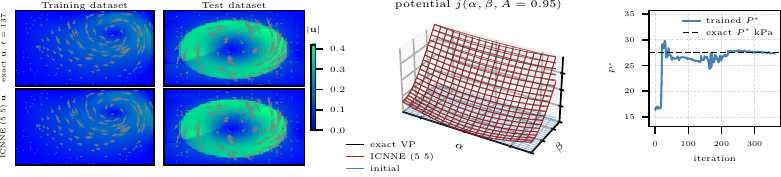}
    \caption{Compressible VP:
    The four  panels to the left compare the exact (top row) and learned (bottom row) velocity fields for the training problem at $t = \SI{137}{\hour}$ (left column) and the test problem (right column). In the middle, we show the learned dissipation potential, and on the right the ice strength $p^\ast$ convergence over training iterations for the best scoring model.
    }
    \label{fig:vp_results}
\end{figure}

\subsection{Learning a sea ice constitutive law from DEM data}\label{subsec:dem}

As a final test case, we learn a constitutive law from data produced with a DEM for sea ice. Unlike the previous experiments, there is no known ground truth constitutive law. The training data is produced with SubZero \cite{manucharyan2022}, a discrete-element model that represents sea ice as a collection of rigid polygonal floes whose dynamics are governed by Newton's laws of motion. Simulations are run on a doubly-periodic $100\,\mathrm{km}\times 100\,\mathrm{km}$ square patch of ocean containing 5000 floes of thickness $\SI{2}{\meter}$ at a mean concentration of $0.8$. The ice is driven only by the ocean current shown on the left of \cref{fig:subzero_results} (there is no wind forcing), and stationary states are obtained for three maximum ocean speeds: $0.1$, $0.5$, and $\SI{1}{\meter\per\second}$. For each stationary state, the floe velocities are area-averaged onto a regular Cartesian grid of size $10\times10$ to produce a gridded velocity field $\bu_{\rm data}$ together with a gridded concentration field $A_{\rm data}$ (the fractional area covered by ice in each cell). To solve the equations of motion described in \cref{subsubsec:seaice}, we set the drag parameters in \eqref{eq:drag_forces} to the values used in SubZero, $C_o = 3\times10^{-3}$, $C_a = 10^{-3}$, $\rho_o = \SI{1027}{\kilo\gram\per\meter\cubed}$, and $\rho_a = \SI{1.2}{\kilo\gram\per\meter\cubed}$. As in \cref{subsec:vp}, the equations are non-dimensionalized with the same values, and the pressure law \eqref{eq:pressure_vp} uses $C = 20$ and $H_0 = \SI{1}{\meter}$. For PDE-constrained training, we discretize the domain with a doubly-periodic $50 \times 50$ quadrilateral mesh that is scaled to match the physical dimensions of the DEM domain. Velocity is approximated with continuous piecewise bilinear elements (CG1) and concentration with piecewise constants (DG0), and we use $\epsilon = 10^{-6}$ in the regularized invariant $\beta_\epsilon$. The DEM velocity and concentration are projected onto the FEM function spaces. Before projection, the concentration perturbation is smoothed by retaining only its leading spectral mode to suppress sub-grid fluctuations. The ocean forcing velocity is similarly interpolated onto the CG1 space.

\paragraph{Training problem (DEM stationary states)}
The training data consists of the three velocity snapshots corresponding to the stationary states at different maximum velocities. Following \cref{subsec:vp}, we assume the pressure law \eqref{eq:pressure_vp} with $C$ and $H_0$ fixed, and we fit the network parameters $\theta$ and the ice strength $p^\ast$. The loss functional is restricted to the FEM cells that coincide with active DEM grid cells, i.e.\ cells for which a DEM measurement is available. Since the DEM output is inherently noisy, no additional noise is added to the training data. The initialization procedure is identical to the VP experiment: the potential is pre-trained to a concentration-dependent quadratic function. There is no test dataset; the quality of each model is evaluated by the final training loss $\Lc(\theta)$. The penalization parameter is set to $\lambda_p = 10^{2}$.

\paragraph{Results}
In \cref{tab:subzero_comparison}, we present the optimal values of the loss function $\Lc(\theta)$; the best fit is achieved with the ICNNE (10 10) architecture. The potential computed with this model is plotted in the right panel of \cref{fig:subzero_results}; interestingly, the potential presents much larger values for $\alpha < 0$, indicating resistance to shear is much larger under compressive rather than expansive strain-rates, something to be expected in a DEM where particles do not stick. In the middle panels of \cref{fig:subzero_results}, we compare the velocities computed with this model with the data.

\begin{table}[tbh]
    \centering
    \caption{DEM: Values of the loss function $\Lc(\theta)$ after training the dissipation potential on the DEM data. For each model, we show the network architecture that achieves the lowest value.}
    \label{tab:subzero_comparison}
    \begin{tabular}{l c c c c c }
    \toprule
     & \textbf{ICNNE} & ICNN & SNN & SNNE & SNNsq \\
    \midrule
    Loss & \textbf{\num{4.4e-05}} & \num{7.9e-05} & \num{8.7e-04} & \num{1.1e-04} & \num{7.5e-05} \\
    Best network & \textbf{10 10} & 5 5 5 & 5 5 5 & 20 & 5 5 5 \\
    \bottomrule
    \end{tabular}
\end{table}

\begin{figure}[tbh]
    \centering
    \includegraphics[width=0.98\textwidth]{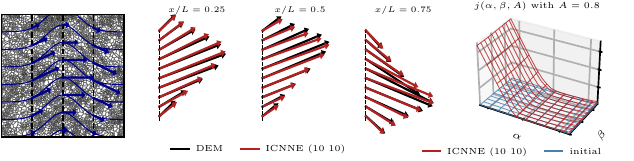}
    \caption{DEM: We learn a constitutive law that reproduces DEM data. In this case, no ground truth constitutive law is known. The panel to the left shows the polygonal ice floes simulated in the DEM and the driving ocean current in blue. The three panels in the middle compare the DEM and model velocity fields at three cross-sections of the domain, for $u_o = \SI{0.5}{\meter\per\second}$. The panel to the right shows the trained dissipation potential (red) together with the initial guess (blue).}
    \label{fig:subzero_results}
\end{figure}

\section{Conclusions}
We developed a robust method that combines flexible NN representations with physical principles to identify constitutive laws in complex compressible and incompressible fluids from one or a few sparse velocity-field snapshots. Rather than learning the constitutive law, which makes it hard to enforce the physical properties required for well-posed governing equations, we learn the dissipation potential, which implicitly defines the constitutive law.
By representing this potential with an even input-convex NN (ICNNE), we incorporate fundamental physical principles which ensure that the governing equations remain solvable during and after training.
We compare several alternative NN parameterizations for the constitutive law and demonstrate that, with these other formulations, maintaining both robustness and solvability is difficult. Using three benchmark complex fluids with known constitutive relations, two compressible and one incompressible, we verify the accuracy of our method and show that alternative approaches either lose robustness or suffer from reduced accuracy.
We also demonstrate the method on a system with an unknown constitutive law, where interactions among irregularly shaped particles produce complex flow.

Regarding computational cost, the training is dominated by the forward and adjoint PDE solves, whose number is independent of the network size and which are independent across snapshots, so the approach scales to larger meshes and to 3D velocity data as long as the forward problem can be solved efficiently. Possible extensions include learning from time-dependent data with inertia, anisotropic constitutive laws using structure tensors, history-dependent laws via internal variables as in \cite{RajCaoStuartEtAl26}, and, motivated by the VP and DEM examples, the inference of sea-ice rheologies from observed ice velocities.

\bibliographystyle{siamplain}
\bibliography{bibliography}

\appendix

\section{Characterization of frame-indifferent constitutive laws}\label{subapp:frame_indifference}
Let $\Sbb_2$ and $\Obb_2$ denote the set of real-valued symmetric and orthogonal matrices in $\RR^{2\times 2}$, with
\begin{align}
	\Obb_2 := \left\lbrace \bO\in\RR^{2\times 2} : \bO\bO^\top = \bO^\top\bO = \bI \right\rbrace.
\end{align}
A matrix operator $\Cc:\Sbb_2\to\Sbb_2$ is said to be rotationally-invariant if
\begin{align}\label{eq:isotropy}
	\Cc(\bO\bA\bO^\top) = \bO \Cc(\bA)\bO^\top
\end{align}
holds for all $\bO\in\Obb_2$. A rotationally-invariant matrix operator yields a local tensor operator that is \emph{frame-indifferent} in the following sense. Under a rotation of the reference frame, the output of the tensor operator transforms accordingly. The following theorem characterizes rotationally-invariant matrix operators in terms of scalar-valued functions.

\begin{theorem}\label{thm:isotropy}%
    A function $\Cc:\Sbb_2\to\Sbb_2$ is rotationally-invariant if and only if there exist two scalar functions $\psi_1$ and $\psi_2$ from $\RR^2$ to $\RR$ such that
    \begin{align}\label{eq:representation_of_C_psi}
    	\Cc(\bA) = \psi_1(\iota_{\bA}) \bI + \psi_2(\iota_{\bA}) \bA,
    \end{align}
    for all $\bA\in\Sbb_2$, with
    $\iota_{\bA}:= \left( \tr{\bA}, \det{\bA} \right)$
    denoting the principal invariants of $\bA$.
\end{theorem}
A proof of \cref{thm:isotropy} can be found in \cite[Section 37]{gurtin1982}. To derive expression \eqref{eq:representation_of_C}, we first note that, for a fixed $\bA\in\Sbb_2$, the map from $\iota_{\bA}$ to $(\tr{\bA},|\dev{\bA}|)$ is bijective from $\RR^2$ into $\RR\times\RRp$, with
\begin{align}
    |\bA|^2:= \frac{1}2\bA:\bA.
\end{align}
One can show that the mapping is bijective via the identity
\begin{align}
    |\dev{\bA}|^2 = \frac{1}{4}(\tr{\bA})^2 - \det{\bA}.
\end{align}
Therefore, we can think of $\psi_1$ and $\psi_2$ as functions of $(\tr{\bA},|\dev{\bA}|)$. By setting $\bA = \bD\bu$ for a velocity field $\bu$ and using the notation from \eqref{eq:invariants}, we define the bulk and shear viscosities as
\begin{align}
    \zeta(\alpha,\beta) := \frac{1}{\alpha}\psi_1(\alpha,\beta) + \frac{1}{2} \psi_2(\alpha,\beta) \quad \text{and} \quad \eta(\alpha,\beta) := \frac{1}{2}\psi_2(\alpha,\beta),
\end{align}
and the characterization \eqref{eq:isotropy} is established. The bulk viscosity $\zeta$, however, may not be well-defined when $\alpha = 0$. For $\zeta$ to be well-defined on all $\RR\times\RRp$, it is necessary that $\lim_{\alpha\to0} \psi_1(\alpha,\beta)/\alpha$ exists in $\RR$ for all $\beta$. Under this additional assumption, we can characterize any rotationally-invariant constitutive law as in \eqref{eq:isotropy}, in terms of a bulk and shear viscosity.

A constitutive law $\Cc=\Cc(\bD\bu)$ defined in terms of a rotationally-invariant matrix operator $\Cc:\Sbb_2\to\Sbb_2$ is said to be isotropic because the material exhibits no preferential directions of deformation. Anisotropy can be introduced without destroying frame-indifference by including additional fabric or structure tensors that track preferential directions of deformation \cite{oda1989, wilchinsky2006, li2009}. In this work, however, we restrict our attention to isotropic constitutive laws.

\end{document}